\documentclass[11pt,letterpaper]{article}

\usepackage{url}
\usepackage{fullpage}

\usepackage{amsmath}
\usepackage{algorithm}
\usepackage[noend]{algpseudocode}
\usepackage{amssymb}
\usepackage{amsthm}
\usepackage{color}
\usepackage{array}
\usepackage{xy}
\usepackage{setspace}
\usepackage{varwidth}
\usepackage{multicol}
\usepackage{algorithm}
\usepackage{hyperref}
\usepackage{thm-restate}
\usepackage{todonotes}

\newcommand{\poly}{\operatorname{poly}}

 \newcommand{\E}{\mathbb{E}}

\newtheorem{thm}{Theorem}[section]

\theoremstyle{remark}

\newtheorem{theorem}{Theorem}[section]

\newtheorem{lemma}[thm]{Lemma}

\theoremstyle{remark}

\renewcommand{\paragraph}[1]{\vspace{.2 cm} \noindent \textbf{#1}}

\usepackage{authblk}

\renewcommand{\epsilon}{\varepsilon}

\begin{document}
\title{An Analysis of Brent's Insertion Method for Hash Tables}
\author{William Kuszmaul}
\affil{Carnegie Mellon University}
\date{}
\maketitle

\begin{abstract}
    In 1968, Richard P.~Brent introduced a new way of building a hash table that, at least empirically, achieves a remarkable property: Even if the hash table is filled to 100\% full, the expected time to query a \emph{random key out of those present} is $O(1)$.
    
    Despite the simplicity of Brent's method, the guarantees of the method have never been formally analyzed. This is due to the subtle issue of handling \emph{spoiled randomness}. The algorithm will sometimes try to use hash functions $h_j$ on keys $y$ that it has already probed in the past. When the algorithm does this, we cannot treat the hash function as random, because its random bits have already affected the state of the table. This issue makes Brent's method surprisingly subtle to reason about formally.

    In this note, we give a simple and formal analysis of Brent's hash table. The analysis can be taught in a graduate randomized algorithms course, and provides a nice example of how to deal with subtle issues in a probabilistic analysis (namely, the issue of spoiled randomness).
\end{abstract}

\section{Introduction}

In 1968, Richard P.~Brent introduced a new way of building a hash table that, at least empirically, achieves a remarkable property: Even if the hash table is filled to 100\% full, the expected time to query a \emph{random key out of those present} is $O(1)$ \cite{Brent68}. 

Brent's method works as follows. Each key $x$ hashes to a \emph{probe sequence} $h_1(x),h_2(x), \ldots \in [n]$ of independent and uniformly random positions in an array $A$ of size $n$. As in a standard random probing hash table \cite{KnuthVol3}, a query for a key $x$ simply examines the positions $h_1(x), h_2(x), \ldots$ in order, until it finds the key.\footnote{If a key is not present, the query can complete once it encounters a free slot in the sequence.} What makes Brent's method interesting is how it performs insertions. The insertion of a key $x$ proceeds in rounds, where round $i$ tries to either insert $x$ directly in position $h_i(x)$ (if the position is free) or to displace some other key $y$ from one of positions $h_j(x) \in \{h_1(x), h_2(x), \ldots, h_{i-1}(x)\}$ and to then move $y$ forward in its own probe sequence by $i - j$ positions (if that position is free). See Algorithm \ref{alg:brent-insert}.

\begin{algorithm}
\caption{Brent's Insertion Method}
\label{alg:brent-insert}
\begin{algorithmic}[1]
\Procedure{Insert}{$x$}
\For{$i = 1, 2, \ldots$}\Comment{Round $i$}
    \If{$A[h_i(x)]$ is empty}
        \State Insert $x$ in position $h_i(x)$
        \State \Return
    \EndIf
    \For{$j = 1, 2, \ldots, i-1$}
        \State $y \gets A[h_j(x)]$
        \State Let $k$ be the smallest $k$ such that $h_k(y) = h_j(x)$ \Comment{$y$'s current position is $h_k(y)$}
        \If{$A[h_{k +i-j}(y)]$ is empty}
            \State Move $y$ to position $h_{k + i-j}(y)$
            \State Insert $x$ in position $h_j(x)$
            \State \Return
        \EndIf
    \EndFor
\EndFor
\EndProcedure
\end{algorithmic}
\end{algorithm}

Brent was interested in the \emph{amortized expected query time} of the hash table, which is the time to query a \emph{random key out of those present}, after the hash table has been filled to be a $1 - \epsilon$ fraction full for some parameter $\epsilon$. Whereas most hash tables have an amortized expected query time that is a function of $\epsilon^{-1}$, Brent hypothesized that his method should have an amortized expected query time of $O(1)$, \emph{regardless of the parameter $\epsilon$}.

Brent's hypothesis, although supported by both experimental and heuristical evidence \cite{Brent68,RobinHood86,GonnetMunro79}, is surprisingly subtle to analyze formally. Despite a great deal of additional work on the topic, including numerous extensions of Brent's technique \cite{GonnetMunro79,Mallach77,Lyon78,Poblete77}, and even a survey on the topic \cite{MunroCelis85}, the basic task of analyzing Brent's method has remained open \cite{MunroCelis85}.

In this note, we give a formal proof that Brent's method achieves amortized expected query time $O(1)$. 

\begin{theorem}
Let $\epsilon \in (0, 1/2)$. Suppose Brent's method is used to insert $(1 - \epsilon)n$ keys into a hash table of size $n$. Then the expected time to query a random key out of those present is $O(1)$. Moreover, if the hash table is not yet full, then the expected time to perform an additional insertion is $O(\epsilon^{-1})$.
\label{thm:main}
\end{theorem}

In addition to bounding the amortized expected query time,  Theorem \ref{thm:main} also obtains a tight bound on expected query time, which resolves a conjecture by Celis and Munro \cite{MunroCelis85} that the expected insertion time is $O(1)$ at load factors $1 - \Omega(1)$. Interestingly, our insertion time bound also offers somewhat of a counterbalance to Knuth \cite{KnuthVol3}, who describes Brent's method as ``doing more work when inserting an item, moving records in order to reduce the expected retrieval time.'' Theorem \ref{thm:main} tells us that, at least asymptotically, the insertion time for Brent's method is the same as that for classic random probing. 

The proof of Theorem \ref{thm:main} is designed to be teachable as a single lecture in a graduate Randomized Algorithms course. It serves as a good example of how to deal with subtle issues in a probabilistic analysis (namely, the issue of \emph{spoiled randomness}), while also showcasing a simple and beautiful algorithm from the early data-structures literature. 

In the decades since Brent's original paper, there has been a huge body of work on how to build time- and space-efficient hash tables (see, e.g., \cite{fiat1988non,fiat1993implicit,RamanRamanRao03,FotakisPaghSandersSpirakis05,DietzfelbingerWeidling07,demaine2006dictionariis,ArbitmanNaorSegev10,bender2022linear,BenderFarachColtonKuszmaul22,IcebergHashing23,bercea2023dynamic,LiLiangYuZhou23,LiLiangYuZhou24,braverman2024tight,bender2025optimal}), leading ultimately to the development of matching upper and lower bounds for the optimal time-space tradeoffs achievable by \emph{any} abstract data structure that solves the hash table problem \cite{BenderFarachColtonKuszmaul22,LiLiangYuZhou23,LiLiangYuZhou24}. For open-addressing hash tables, it is only in the last few years that we have seen examples of hash table analyses that formally achieve amortized expected query time $O(1)$ \cite{BenderKuszmaulZhou24,FarachColtonKrapivinKuszmaul24}. More than half a century after Brent's contribution, it remains the case that his method is by far the simplest known technique for achieving the guarantee. The results in this paper show that this simple algorithm also has a relatively clean and simple analysis. 

\section{The Analysis}

In this section, we prove Theorem \ref{thm:main}. Before diving into the proof, let us take a moment to first understand why, \emph{intuitively}, Brent's method should achieve amortized expected query time $O(1)$.

Consider an insertion taking place when the hash table is a $1- \epsilon$ fraction full. In the first $i$ rounds, the insertion probes $\Theta(i^2)$ different positions to check if they are free. If we imagine that each of these probes has an $\epsilon$ probability of success, then it stands to reason that the insertion should complete after roughly $O(\sqrt{\epsilon^{-1}})$ rounds. By design, if the insertion completes in the $r$-th round, for some $r$, then it increases the \emph{sum of the query times of all elements present} by exactly $r$. This suggests that the increase to query times from an insertion at load factor $1 - \epsilon$ is roughly $O(\sqrt{\epsilon^{-1}})$. Integrating over insertions, even if we fill the hash table to 100\% full, the amortized expected query time should be roughly 
$$\int_0^{1} \sqrt{\epsilon^{-1}} \, d\epsilon = O(1).$$

What makes this intuition tricky to formalize is the issue of \emph{spoiled randomness}. For each key $y$ in the hash table, and each hash function $h_j$, the \emph{first time} that we ever probe $h_j(y)$ we can think of the hash function as providing ``fresh randomness''. If, at that point in time, the hash table is $1 - \epsilon$ full, then the probe has an $\epsilon$ probability of finding a free slot. Notice, however, that for the same key $y$, there may be \emph{many insertions} that consider moving $y$ to position $h_j(y)$. After the first such attempt, we can no longer think of subsequent probes to position $h_j(y)$ as examining a random slot. Rather, such probes are ``wasted'' on a slot that we already know to be occupied. 

As a convention, for a given key $z$ and hash function $h_i$, we say that $h_i$ is \emph{fresh} for key $z$ if it has never yet been probed by any insertion, and is \emph{spoiled} for key $z$ if some past insertion has already probed it. Likewise, when an insertion probes $h_i(z)$ for some key $z$ and hash function $h_i$, we say that the probe is \emph{fresh} if $h_i$ is fresh for $z$, and that the probe is \emph{spoiled} otherwise. Spoiled probes are guaranteed to fail, since (unbeknownst to the current insertion) the same probe has already been performed in the past, and the slot being probed is guaranteed to be occupied.

To bound the effect of spoiled probes on the analysis, we make use of the following result, which is a standard generalization of the classic coupon-collector lemma to the case where we wish to collect only a $(1 - \epsilon)$ fraction of coupons (see, e.g., \cite{FarachColtonKrapivinKuszmaul24} for a simple proof).

\begin{lemma}
\label{lem:coupon-collector}
Let $n \in \mathbb{N}$ and let $1/n \le \epsilon \le 1/2$. Let $r_1, r_2, \ldots$ be independent and uniformly random positions in $[n]$. Let $T$ be the smallest index such that $$\left|\bigcup_{i = 1}^T \{r_i\}\right| \ge (1 - \epsilon)n.$$ Then, with probability $1 - 1 / \poly(n)$, we have $T = \Theta(n \log \epsilon^{-1})$.
\end{lemma}

Intuitively, the coupon-collector lemma can be used to bound the number of pairs $(h_j, y)$ such that $h_j$ is already spoiled for $y$. It tells us that, when the hash table is $1 - \epsilon$ full, the total number of such pairs is $O(n \log \epsilon^{-1})$. This means that, for a \emph{random} element $y$ in the table, $y$ is only involved in $O(\log \epsilon^{-1})$ such pairs in expectation. Thus, on a given insertion, and for a given element $y$ that the insertion displaces, we only expect to incur roughly $O(\log \epsilon^{-1})$ spoiled probes on that element. This basic insight is the starting point of the analysis, allowing us to prove the following lemma, which bounds the number of rounds that a given insertion takes. 
\begin{lemma}
Suppose that after inserting $(1 - \epsilon)n$ elements $z_1, z_2, \ldots, z_{(1 - \epsilon)n}$, we insert an additional element $z$. For all $j \ge \sqrt{\epsilon^{-1}}$, the probability that the insertion requires at least $j$ rounds is at most $1/e^{\Omega(\epsilon j^2)} + 1/\poly(n)$.
\label{lem:rounds}
\end{lemma}
\begin{proof}
Let $P(z_i)$ denote the number of hash functions $h_1, h_2, \ldots$ that are already spoiled for $z_i$ (they have already been probed by at least one past insertion). Then, by the generalized coupon-collector lemma (Lemma \ref{lem:coupon-collector}), with probability $1 - 1 / \poly(n)$, we have 
\begin{equation}
\sum_{i = 1}^{(1 - \epsilon)n} P(z_i) = O(n \log \epsilon^{-1}).
\label{eq:coupon-collector-sum}
\end{equation}
For the rest of the proof, let us condition on some fixed outcome of the first $(1 - \epsilon)n$ insertions such that \eqref{eq:coupon-collector-sum} holds. The only randomness that we will use in our analysis is the randomness of the hashes $h_i(z)$ for the element $z$ currently being inserted, and the randomness of fresh probes (i.e., of $h_i(z_j)$ for hash functions $h_i$ and keys $z_j$ such that $h_i$ is not yet spoiled for $z_j$). 

Now, consider the elements $u_1, u_2, \ldots, u_j$ in positions $h_1(z), h_2(z), \ldots, h_j(z)$ that we consider displacing during the $j$-th round of the insertion. If position $h_j$ is empty, then define $u_j = \text{null}$ and say as a convention that $P(u_j) = 0$. Because $h_1, h_2, \ldots, h_j$ are independent random hash functions (and because we have already conditioned on outcomes for $P(z_i)$ for all $i$), the random variables $P(u_1), P(u_2), \ldots, P(u_j)$ are independent. Moreover, by \eqref{eq:coupon-collector-sum}, each $P(u_j)$ has expected value at most $c \log \epsilon^{-1}$ for some positive constant $c$.  By Markov's inequality, it follows that $\Pr[P(u_i) \ge 2c \log \epsilon^{-1}] \le 1/2$. Thus, each $P(u_i)$ is independently at most $2c \log \epsilon^{-1}$ with probability at least $1/2$.

Recall by a Chernoff bound that, if we flip $k$ independent coins, each of which has probability at least $1/2$ of being heads, then with probability at least $1 - 1/2^{\Omega(k)}$, at least $k/4$ of the coins are heads. Applying this to the random variables $P(u_1), P(u_2), \ldots, P(u_j)$, where a coin flip corresponds to the event $P(u_i) \ge 2c \log \epsilon^{-1}$, we can conclude that, with probability at least 
\begin{equation}
1 - 1/2^{\Omega(j)},
\label{eq:chernoff-bound}
\end{equation}
at least $j/4$ of the values $P(u_1), \ldots, P(u_j)$ are at most $2c\log \epsilon^{-1}$. Condition on this event, and label these $j/4$ elements by $u_1', u_2', \ldots, u_{j/4}'$.

If any of $u_1', u_2', \ldots, u_{j/4}'$ are null (i.e., they correspond to an empty position), then the insertion will trivially succeed by the end of the $j$-th round. Otherwise, $u_1', u_2', \ldots, u_{j/4}'$ are each elements whose hash functions $\{h_i \mid i > 2c \log \epsilon^{-1}\}$ have never yet been probed. The first $j$ rounds of the insertion are therefore able to make at least 
$$\sum_{s = 1}^{j/4} (s - 2 c \log \epsilon^{-1}) = \Omega(j^2) - O(j \log \epsilon^{-1})$$
\emph{fresh} probes. Since $j \ge \omega(\log \epsilon^{-1})$, this amounts to at least $\Omega(j^2)$ fresh probes. Each of these probes independently has probability $\epsilon$ of finding a free slot, so the probability of the round failing is at most 
\begin{equation}
(1 - \epsilon)^{\Omega(j^2)} = 1/e^{\Omega(\epsilon j^2)}.
\label{eq:ind}
\end{equation}

Putting the pieces together, the overall probability of the insertion failing to complete after $j$ rounds is (from \eqref{eq:coupon-collector-sum}, \eqref{eq:chernoff-bound}, and \eqref{eq:ind}) at most 
$$1/\poly(n) +1/e^{\Omega(j)} + 1/e^{\Omega(\epsilon j^2)} = 1/e^{\epsilon \Omega(j)} + 1/\poly(n),$$
where the $1/\poly(n)$ term comes from the probability of the coupon-collector lemma failing, the $1/e^{\Omega(j)}$ term comes from \eqref{eq:chernoff-bound}, and the $1/e^{\Omega(\epsilon j^2)}$ term comes from the final step of the analysis. 
\end{proof}

With the help of Lemma \ref{lem:rounds}, we can prove Theorem \ref{thm:main} as follows.

\begin{proof}[Proof of Theorem \ref{thm:main}]
We begin by analyzing an insertion of an element $z$ that takes place when the hash table is $1 - \epsilon$ full. By Lemma \ref{lem:rounds}, the probability that the insertion requires at least $j$ rounds is at most $1/e^{\Omega(\epsilon j^2)} + 1/\poly(n)$. The expected number of rounds that the insertion takes is therefore $O(\sqrt{\epsilon^{-1}})$, while the expected number of probes that the insertion makes is $O(\epsilon^{-1})$ (since the $j$-th round makes $O(j^2)$ probes). To bound the insertion time, we must also include the time to calculate during each round $j$ the current hash function $h_k$ that the element in position $h_j(z)$ is using. Because $A[h_j(z)]$ (if it is not empty) is a random element in the hash table, this is the same time as the expected time to query a random element out of those present, which we will shortly show to be $O(1)$. Thus the overall expected insertion time is $O(\epsilon^{-1})$. 

Now let us consider the expected time to query a random key out of those present, after the hash table has had $(1 - \epsilon)n$ insertions. If the $i$-th insertion takes $r_i$ rounds, then it increases the sum of the query times of all elements present by exactly $r_i$. The expected sum of the query times is therefore 
$$\E\left[\sum_{i = 1}^{(1 - \epsilon)n} r_i\right].$$ 
Since an insertion at $1- \delta$ full takes $O(\delta^{-1/2})$ expected rounds, this is at most 
$$O\left(\sum_{i = 1}^{(1 - \epsilon)n} \sqrt{\frac{n}{n - i}} \right) = O\left(n \int_{\epsilon}^{1} \delta^{-1/2} d\delta\right) = O(n).$$
The expected time to query a random key out of those present is therefore $O(1)$.

\end{proof}

\section*{Acknowledgments}
This work was supported in part by NSF grants CCF-2542165 and CCF-2504471, and by a Jane Street grant.

\bibliographystyle{alpha}
\bibliography{references}

\end{document}